\documentclass[reqno]{amsart}

\usepackage{booktabs} 
\usepackage{amssymb,latexsym,amsfonts,amsmath,thmtools}

\usepackage{graphicx}
\usepackage{enumerate}
\usepackage{multirow}
\usepackage{mathtools}
\usepackage{subfigure}
\newcommand{\pre}[1]{\prescript{#1}{}}
\usepackage[export]{adjustbox}
\usepackage{framed}
\usepackage{tcolorbox}
\usepackage{chngcntr}
\usepackage[normalem]{ulem}

\usepackage{tikz}
\usepackage{bm}
\def\baselinestretch{1}
\usepackage{dsfont}
\usepackage[bb=boondox]{mathalfa}
\usepackage{subfig}

\newtheorem{theorem}{Theorem}
\newtheorem{definition}[theorem]{Definition}

\newtheorem{lemma}[theorem]{Lemma}
\newtheorem{assumption}[theorem]{Assumption}
\newtheorem{proposition}[theorem]{Proposition}
\newtheorem{remark}[theorem]{Remark}

\newcommand{\II}{\mathcal{I}_d}

\newcommand{\R}{{\mathbb{R}}}

\newcommand{\TF}{\mathcal{F}}

\newcommand{\N}{{\mathbb{N}}}

\newcommand{\op}{{\mathcal H}}
\newcommand{\SF}{\mathcal{S}}

\newcommand{\Let}{:=}

\def\bb#1{{\textcolor{black}{{\bf}#1}}}

\definecolor{myco}{rgb}{0.2660 0.5740 0.1980}

\newcommand{\Z}{\mathbb{Z}}

\begin{document}
	
	\title[Symbolic Models for a Class of Impulsive Systems]{Symbolic Models for A Class of Impulsive Systems}

	\author{Abdalla Swikir$^1$}
	\author{Antoine Girard$^2$}
	\author{Majid Zamani$^{3,4}$}
	\address{$^1$Department of Electrical and Computer Engineering, Technical University of Munich, Germany.}
	\email{abdalla.swikir@tum.de}
	\address{$^2$Universit\'e Paris-Saclay, CNRS, CentraleSup\'elec, Laboratoire des signaux et syst\`emes
		91190, Gif-sur-Yvette, France.}
	\email{Antoine.Girard@l2s.centralesupelec.fr}
	\address{$^3$Computer Science Department, University of Colorado Boulder, USA.}
	\address{$^4$Computer Science Department, Ludwig Maximilian University of Munich, Germany.}
	\email{majid.zamani@colorado.edu}
	\maketitle

	\begin{abstract} 
		Symbolic models have been used as the basis of a systematic framework to address control design of several classes of hybrid systems with sophisticated control objectives. However, results available in the literature are not concerned with impulsive systems which are an important modeling framework of many applications.   
		In this paper, we provide an approach for constructing symbolic models for \bb{a class of} impulsive systems possessing some stability properties. We formally relate impulsive systems and their symbolic models using a notion of so-called alternating simulation function. We show that behaviors of the constructed symbolic models are approximately equivalent to those of the impulsive systems.
		Finally, we illustrate the effectiveness of our results through a model of storage-delivery process by constructing its symbolic model and designing controllers enforcing some safety specifications. 
	\end{abstract}
	\vspace{-0.3cm}
	\section{Introduction}
	Symbolic models have been the \bb{aim} of intensive study in the last two decades since they provide a mechanism for reducing complexity in the analysis and control of cyber-physical systems \cite{Tabu, Belta}. They serve as abstract mathematical models where each symbolic state and input represent a collection of continuous states and inputs in the original concrete model. As they have finite number of states and inputs, they enable the use of  correct-by-construction methods from the computer science community to design controllers
	for a wide variety of systems. For instance, they allow one to use automata-theoretic methods \cite{MalerPnueliSifakis95} to design controllers for hybrid systems with respect to logic specifications such as those expressed as linear temporal logic (LTL) formulae \cite{Katoen}.  
	In such frameworks, controllers designed for symbolic models can be refined to ones for concrete systems based on some behavioral relation between original systems and their symbolic models such as approximate alternating simulation relations \cite{pt09} or feedback refinement relations \cite{7519063}.
	
	The synthesis of symbolic models for different classes of systems has been investigated, among many others, in the following papers: for incrementally stable and incrementally forward complete nonlinear control systems in \cite{POLA20082508} and \cite{6082386}, respectively; for nonlinear switched systems in \cite{Girard,zamani2015symbolic,SAOUD201858}; for \bb{nonlinear control systems with known constant time delays and time-varying delays} in \cite{pola2010,pola2015}; for networked control systems in \cite{8355500,8010291}, and finally for incrementally stable infinite-dimensional systems with finite-dimensional input spaces in \cite{Girard2014,JZ3}. 
	All the aforementioned approaches essentially take a monolithic view of the systems while constructing symbolic models.
	On the other hand, different compositional methods for constructing symbolic models have been recently introduced in the literature with or without imposing stability assumptions over the network; see \cite{meyer,7403879,SWIKIR2019,8728138,SWIKIR2019a,SWIKIR2019d} and references therein. Although the literature on symbolic models is very rich, unfortunately, there are no results so far on constructing symbolic models for impulsive systems.
	\vspace{-0.08cm}
	
	Impulsive systems are an important class of hybrid systems that contain
	discontinuities or jumps (also referred to as impulses) in the state and input trajectories of the
	system governed by discrete dynamics \cite{4806347,Haddad}. They serve as an important modeling framework for a very large variety of applications, e.g. power electronics, sample-data systems, bursting rhythm models in
	medicine, and some models in economics; see \cite{Miller, Liu} and the references therein. Hence, constructing symbolic models for impulsive systems enlarges the class of systems for which designing correct-by-construction controllers enforcing complex logic specifications is possible. 
	\bb{In this work, we consider time-dependent impulsive systems in which the distance between the impulses is assumed to belong to a finite set. Such a class of systems is well studied in the literature; see \cite{Rios} and references therein. For example, this class of systems models the dynamics of the estimation error in networked control systems \cite[Section 8.2.]{Teel} while assuming time instants of the reception of measurements are nondetermined but lie in a finite set.}

	This paper provides for the first time an approach for synthesizing symbolic models for \bb{a class of} impulsive systems. The symbolic models constructed in this work are complete as their behaviors are approximately equivalent to those of the concrete systems \cite{Tabu}. First, we introduce a class
	of transition systems which allows us to model impulsive systems and their
	symbolic models in a common framework. Then we recall a notion of so-called alternating simulation function to relate two transition systems. 
	Such a function allows one to determine quantitatively the
	mismatch between the observed behavior of two systems, and implies the existence of an approximate alternating simulation relation between them \cite{pt09}. Second, we provide a methodology for constructing symbolic models together with their alternating simulation functions for impulsive systems possessing some incremental stability properties. In particular, we require that either the continuous or the discrete dynamic of the impulsive system to be incrementally input-to-state stable \cite{angeli09} while the other one is forward complete \cite{6082386}. Given such an incremental property, we show that the constructed symbolic model is indeed a complete one \cite{Tabu} (cf. Remark \ref{biss}).
	Finally, we apply our results to a model of storage-delivery process by constructing its symbolic model under different stability properties. We also design a controller maintaining the number of items in the storage in a desired range.
	
	\section{Notation and Preliminaries}\label{1:II}
	\subsection{Notation}\label{note}
	We denote by $\R$, $\Z$, and $\N$ the set of real numbers, integers, and non-negative integers,  respectively.
	These symbols are annotated with subscripts to restrict them in
	the obvious way, e.g., $\R_{>0}$ denotes the positive real numbers. We denote the closed, open, and half-open intervals in $\R$ by $[a,b]$,
	$(a,b)$, $[a,b)$, and $(a,b]$, respectively. For $a,b\in\N$ and $a\le b$, we
	use $[a;b]$, $(a;b)$, $[a;b)$, and $(a;b]$ to
	denote the corresponding intervals in $\N$.
	Given any $a\in\R$, $\vert a\vert$ denotes the absolute value of $a$. Given any $u=(u_1,\ldots,u_n)\in\R^{n}$, the infinity norm of $u$ is defined by $\Vert u\Vert=\max_{1\leq i\leq n}|u_i|$.  
	Given a function $\nu: \R_{\ge0} \rightarrow \R^n $, the supremum of $ \nu $ is denoted by $ \Vert \nu\Vert_\infty $; we recall that $ \Vert \nu\Vert_\infty := \text{sup}_{t\in\R_{\ge0}}\Vert \nu(t)\Vert$. Given $\mathbf x:\R_{\geq0}\rightarrow\R^{n},\forall t, s\in\R_{\geq0}$ with $t\geq s$, we define $\mathbf{x}(\pre{-}t)=\lim_{s\rightarrow t}\mathbf{x}(s)$.
	We denote by $\text{card}(\cdot)$ the cardinality of a given set and by $\emptyset$ the empty set. \bb{Given sets $X$ and $Y$, we denote by $f:X\rightarrow Y$ an ordinary map of $X$ into $Y$, whereas $f:X\rightrightarrows Y$ denotes a set-valued map \cite{Rock0000}.} For any set \mbox{$S\subseteq\R^n$} of the form $S=\bigcup_{j=1}^MS_j$ for some $M\in\N$, where $S_j=\prod_{i=1}^n [c_i^j,d_i^j]\subseteq \R^n$ with $c^j_i<d^j_i$, and nonnegative constant $\eta\leq\tilde{\eta}$, where $\tilde{\eta}=\min_{j=1,\ldots,M}\eta_{S_j}$ and \mbox{$\eta_{S_j}=\min\{|d_1^j-c_1^j|,\ldots,|d_n^j-c_n^j|\}$}, we define \mbox{$[S]_{\eta}=\{a\in S\,\,|\,\,a_{i}=k_{i}\eta,k_{i}\in\mathbb{Z},i=1,\ldots,n\}$} if $\eta\neq0$, and $[S]_{\eta}=S$ if $\eta=0$. The set $[S]_{\eta}$ will be used as a finite approximation of the set $S$ with precision $\eta\neq0$. Note that $[S]_{\eta}\neq\emptyset$ for any $\eta\leq\tilde{\eta}$. 
	We use notations $\mathcal{K}$ and $\mathcal{K}_\infty$
	to denote different classes of comparison functions, as follows:
	$\mathcal{K}=\{\alpha:\mathbb{R}_{\geq 0} \rightarrow \mathbb{R}_{\geq 0} |$ $ \alpha$ is continuous, strictly increasing, and $\alpha(0)=0\}$; $\mathcal{K}_\infty=\{\alpha \in \mathcal{K} |$ $ \lim\limits_{s \rightarrow \infty} \alpha(s)=\infty\}$.
	For $\alpha,\gamma \in \mathcal{K}_{\infty}$ we write $\alpha\le\gamma$ if $\alpha(r)\le\gamma(r)$, and, by abuse of notation, $\alpha=c$ if $\alpha(r)=cr$ $\forall r\in\R_{\geq0}$. Finally, we denote by $\II$ the identity function over $\R_{\ge0}$, i.e. $\II(r)=r, \forall r\in \R_{\ge0}$.
	
	\subsection{Nonlinear Impulsive Systems} 
	\bb{Among several classes of impulsive systems studied in the literature, e.g., \cite{4806347,Haddad, Teel}, in this work, we study a class of time-dependent nonlinear impulsive systems as defined next.}
	\begin{definition}\label{def:sys1}
		A nonlinear impulsive system $\Sigma$ is defined by the tuple	$\Sigma=(\R^n,\mathbb U,\mathcal{U},f,g)$,
		where $\R^n$ is the state space, $\mathbb U\subseteq\R^m$ is the input set, $\mathcal{U}$ is the set of all measurable bounded input functions $\nu:\R_{\geq0}\rightarrow \mathbb U$, and $f,g: \R^n\times \mathbb U \rightarrow \R^n $ are locally Lipschitz functions; 
		
		The nonlinear impulsive system $\Sigma $ is described by differential and difference equations of the form
		\begin{align}\label{eq:2}
		\Sigma:\left\{
		\begin{array}{rl}
		\mathbf{\dot{x}}(t)&= f(\mathbf{x}(t),\nu(t)),\quad\quad\quad\quad t\in\R_{\geq0}\backslash \Omega,\\
		\mathbf{x}(t)&= g(\mathbf{x}(\pre{-}t),\nu(t)),\quad\quad\quad\,\, t\in \Omega,
		\end{array}
		\right.
		\end{align}where $\Omega=\{t_k\}_{k\in\N}$ with $t_{k+1}-t_{k}\in\{p_1\tau,\ldots,p_2\tau\}$ for fixed jump parameters $\tau\in\R_{>0}$ and $p_1,p_2\in \N_{\ge1}$, $p_1\le p_2$; and, $\mathbf{x}:\R_{\geq0}\rightarrow \R^n $ is the state signal, which is assumed to be right-continuous for all $t\in\R_{\ge0}$, and $\nu\in\mathcal{U}$ is the input signal. We will use $\mathbf{x}_{x,\nu}(t)$ to denote a point reached at time $t\in \R_{\geq0}$ from initial state $x=\mathbf{x}(0)$ under input signal  $\nu\in\mathcal{U}$. \bb{The Lipschitz condition imposed on $f$ ensures the existence and uniqueness of a solution of system $\Sigma$ in \eqref{eq:2}; see \cite{Kulev, Dishliev} for more details.} We denote by $\Sigma_c$ and $\Sigma_d$ the continuous and discrete dynamics of system $\Sigma$, i.e., $\Sigma_c:\mathbf{\dot{x}}(t)= f(\mathbf{x}(t),\nu(t))$, and $\Sigma_d:\mathbf{x}(t)= g(\mathbf{x}(\pre{-}t),\nu(t))$. 
	\end{definition}	
	\section{Transition Systems and Alternating Simulation Functions}\label{I}
	We start by introducing the class
	of transition systems \cite{Tabu} which allows us to model impulsive and
	symbolic systems in a common framework.
	
	\begin{definition}\label{ts} A transition system is a tuple $T=(X,X_0,U,\TF,Y,\op)$ consisting of:
		\begin{itemize}
			\item a set of states $X$;
			\item a set of initial states $X_0\subseteq X$;
			\item a set of inputs $U$;
			\item transition function $\TF: X\times U \rightrightarrows  X$;
			\item an output set $Y$;
			\item an output map $\mathcal{H}:X\rightarrow Y$.
		\end{itemize}
	\end{definition}
	The transition $x^+\in \TF(x,u)$ means that the system can evolve from state $x$ to state $x^+$ under
	the input $u$. Thus, the transition function defines the
	dynamics of the transition system.
	Sets $X$, $U$, and $Y$ are assumed to be subsets of normed vector spaces with appropriate finite dimensions. If for all $x\in  X, u\in  U$, $\text{card}(\TF(x,u))\leq1$ we say that $T$ is deterministic, and non-deterministic otherwise. Additionally, $T$ is called finite if $ X, U$ are finite sets and infinite otherwise. Furthermore, if for all $x\in  X$ 
	there exists $ u\in  U$ such that $\text{card}(\TF(x,u))\neq0$ we say that $T$ is non-blocking. In this work, we only deal with non-blocking transition systems.
	
	Next we introduce a notion of so-called alternating simulation functions, inspired by \cite[Definition $1$]{Girard2009566}, which quantitatively relates two transition systems.
	\begin{definition}\label{sf} 
		Let $T=(X,X_0,U,\TF,Y,\op)$  and $\hat T=(\hat{X},\hat{X}_0,\hat{U},\hat{\TF},\hat{Y},\hat{\op})$ be transition systems with  $\hat Y\subseteq Y$. A function $ \tilde{\mathcal{S}}:X\times \hat X \to \mathbb{R}_{\geq0} $ is called an alternating simulation function from $\hat T$ to $T$
		if there exist $\tilde{\alpha} \in \mathcal{K}_{\infty}$,  $0<\tilde{\sigma}< 1$, $ \tilde{\rho}_{u} \in \mathcal{K}_{\infty}\cup \{0\} $, and some $\tilde{\varepsilon}\in \mathbb{R}_{\geq 0}$ so that the following hold:
		\begin{itemize}
			\item For every $ x\in X,\hat{x}\in\hat{X}$, one has
			\begin{align}\label{sf1}
			\tilde{\alpha} (\Vert\op(x)-\hat{\op}(\hat{x})\Vert) \leq \tilde{\SF}(x,\hat{x}).
			\end{align}
			\item For every $x\in X,\hat{x}\in\hat{X},\hat u\in\hat{U}$, there exists $u\in U$ such that for every $ x^+\in\TF(x,u)$  there exists $\hat{x}^+\in\hat{\TF}(\hat{x},\hat{u})$ so that 
			\begin{align}\label{sf2}
			\tilde{\SF}&(x^+,\hat{x}^+)\leq \tilde{\sigma} \tilde{\SF}(x,\hat{x})+\tilde{\rho}_u(\Vert\hat{u}\Vert_{\infty} )+\tilde{\varepsilon}.
			\end{align}
		\end{itemize}
	\end{definition}
	\bb{The next lemma is adapted from \cite[Theorem $1$]{arxiv} and stated without a proof. This lemma is needed in the proof of Proposition \ref{error}.} 
	\bb{\begin{lemma}\label{sfm} 
			Let $ \tilde{\mathcal{S}}$ be an alternating simulation function from $\hat T$ to $T$ as in Definition \ref{sf}.  
			Then for every $x\in X,\hat{x}\in\hat{X},\hat u\in\hat{U}$, there exists $u\in U$ such that for every $ x^+\in\TF(x,u)$ there exists $\hat{x}^+\in\hat{\TF}(\hat{x},\hat{u})$ so that 
			\begin{align}\label{sf2m}
			\tilde{\SF}&(x^+,\hat{x}^+)\leq \max\{{\sigma}\tilde{\SF}(x,\hat{x}),{\rho}(\Vert\hat{u}\Vert_{\infty} ),{\varepsilon}\},
			\end{align}where ${\sigma}=1-(1-\psi)(1-{\tilde{\sigma}})$, $ {\rho}=\frac{1}{(1-{\tilde{\sigma}})\psi}{\tilde{\rho}_u}$, and ${\varepsilon}=\frac{\tilde{\varepsilon}}{(1-{\tilde{\sigma}})\psi}$ for an arbitrarily chosen positive constant $\psi<1$, and $\tilde{\sigma},\tilde{\rho}_u,\tilde{\varepsilon}$ appearing in Definition \ref{sf}. 
	\end{lemma}}
	
	\bb{Before showing the next result, let us recall the definition of an alternating simulation relation introduced in \cite{pt09}.
		\begin{definition}\label{sr} 
			Let $T=(X,X_0,U,\TF,Y,\op)$  and $\hat T=(\hat{X},\hat{X}_0,\hat{U},\hat{\TF},\hat{Y},\hat{\op})$ be transition systems with  $\hat Y\subseteq Y$. A relation $R\subseteq X\times \hat{X}$ is called an $\hat{\varepsilon}$-approximate alternating simulation relation from $\hat T$ to $T$ if for any $(x,\hat x)\in R$
			\begin{itemize}
				\item $(i)$ $\Vert\op(x)-\hat{\op}(\hat{x})\Vert \leq \hat{\varepsilon}$;
				\item $(ii)$ For ny $\hat u\in\hat{U}$, there exists $u\in U$ such that for all $x^+\in\TF(x,u)$ there exists $ \hat{x}^+\in\hat{\TF}(\hat{x},\hat{u})$ satisfying $(x^+,\hat{x}^+)\in R$
			\end{itemize}
	\end{definition}}
	\bb{In addition, if ($ii$) still holds when reversing the role of $T$ and $\hat T$, The relation $R$ is in fact an $\hat{\varepsilon}$-approximate alternating bisimulation relation between $T$ and $\hat T$ \cite{pt09} (see Remark 4).}
	
	The next result shows that the existence of an alternating simulation function for transition systems implies the existence of an approximate alternating simulation relation between them as as defined above	
	\begin{proposition}\label{error}
		Let $T=(X,X_0,U,\TF,Y,\op)$  and $\hat T=(\hat{X},\hat{X}_0,\hat{U},\hat{\TF},\hat{Y},\hat{\op})$ be transition systems with  $\hat Y\subseteq Y$. Assume $ \tilde{\SF}$ is an alternating simulation function from $\hat T$ to $T$ as in Definition \ref{sf} and that there exists $r\in \R_{>0}$ such that $\Vert \hat{u} \Vert_{\infty} \leq r$ for all $ \hat{u} \in {\hat{U}}$. Then, relation $R\subseteq X\times \hat{X}$ defined by $$R=\left\{(x,\hat{x})\in {X}\times \hat{X}|\tilde{\SF}(x,\hat{x})\leq \max\left\{{\rho}(r),{\varepsilon}\right\}\right\},$$ where $\rho,\varepsilon$ as in Lemma \ref{sfm}, is an $\hat{\varepsilon}$-approximate alternating simulation relation from $\hat T$ to $T$ with 
		$$\hat{\varepsilon}=\tilde{\alpha}^{-1}(\max\{{\rho}(r),{\varepsilon}\}).$$
	\end{proposition}
	\begin{proof}
		Item $(i)$ in Definition \ref{sr} is a simple consequence of the definition of $R$ and condition \eqref{sf1} (i.e. $\tilde{\alpha} (\Vert\op(x)-\hat{\op}(\hat{x})\Vert) \leq \tilde{\SF}(x,\hat{x})\leq\max\{{\rho}(r),{\varepsilon}\}$), which results in $\Vert\op(x)-\hat{\op}(\hat{x})\Vert \leq \tilde{\alpha}^{-1}(\max\{{\rho}(r),{\varepsilon}\})=\hat\varepsilon$. Item $(ii)$ in Definition \ref{sr} follows immediately from the definition of $R$, condition \eqref{sf2m} in Lemma \ref{sfm}, and the fact that $0<{\sigma}<1$. In particular, we have $\tilde{\SF}(x^+,\hat{x}^+)\leq\max\{{\rho}(r),{\varepsilon}\}$ which implies $(x^+,\hat{x}^+)\in R$.
	\end{proof}
	
	The approximate alternating simulation relation guarantees that
	for each output behavior of $T$ there
	exists one of $\hat T$ such that the distance between these output behaviors is uniformly bounded by $\hat{\varepsilon}$. 
	\begin{remark}
		Since the input set in all practical applications is bounded,
		requiring the control inputs to be bounded is not restrictive at all. Moreover, under certain properties of impulsive systems (see Section \ref{1:IV}), one can choose function $\tilde{\rho}_{u}$ the definition of $R$ to be identically zero which cancels the dependency to the size of control inputs in Proposition \ref{error}.
		\hfill$\diamond$
	\end{remark}
	%
	\section{Construction of Symbolic Models}\label{1:IV}
	This section contains the main contribution of this work and its results rely on additional assumption on $\mathcal{U}$ that we now describe.
	Consider impulsive system $\Sigma=(\R^n,\mathbb U,\mathcal{U},f,g)$  with jump parameters $\tau$, $p_1$ and $p_2$. We restrict attention to sampled-data impulsive systems, where input curves belong to $\mathcal{U}_{\tau}$ containing only curves, constant in duration $\tau$, i.e.
	\begin{align}\label{input}
	\mathcal{U}_\tau=\{\nu:\R_{\ge0}\rightarrow \mathbb U| \nu(t)=\nu((k-1)\tau),t\in [(k-1)\tau,k\tau),k\in\N_{\geq1}\}.
	\end{align}
	Next we define sampled-data impulsive systems as a transition system. Such a transition system would be the bridge that relates impulsive systems to their symbolic models.
	\begin{definition}\label{tsm} Given an impulsive system $\Sigma=(\R^n,\mathbb U,\mathcal{U}_{\tau},f,g)$, with jump parameters ($\tau$, $p_1$, $p_2$), we define the associated transition system $T_{\tau}(\Sigma)=(X,X_0,U,\TF,Y,\op)$ 
		where:
		\begin{itemize}\addtolength{\itemindent}{-0.4cm}
			\item  $X=\R^n \times \{0,\ldots,p_2\}$;
			\item  $X_0=\R^n\times \{0\}$; 
			\item $U=\mathcal{U}_{\tau}$;
			\item  $(x^+,l^+)\in \TF((x,l),\nu)$ if and only if one of the following scenarios hold:
			\begin{itemize}
				\addtolength{\itemindent}{-0.85cm}
				\item Flow scenario: $0\leq l\leq p_{2}-1$, $x^+= \mathbf x_{x,\nu}(\pre{-}\tau)$, and $l^+=l+1$;
				\item Jump scenario: $p_1\leq l\leq p_2$,  $x^+= g(x,\nu(0))$, and $l^+=0$; 
			\end{itemize}
			\item $Y=\R^n$;
			\item $\mathcal{H}:X\rightarrow Y$, defined as $\mathcal{H}(x,l)=x$.
		\end{itemize}
	\end{definition}
	
	In order to construct a symbolic model for $T_{\tau}(\Sigma)$, we introduce the following assumptions and lemma. 
	
	\begin{assumption}\label{likeiss}
		Consider impulsive system $\Sigma=(\R^n,\mathbb U,\mathcal{U}_{\tau},f,g)$ with jump parameters $\tau$, $p_1$ and $p_2$. Assume that there exist a locally Lipschitz function $ V:\R^n\times \R^n \to \mathbb{R}_{\geq0} $, $\mathcal{K}_{\infty}$ functions $\underline{\alpha}, \overline{\alpha}, \rho_{u_c},\rho_{u_d}$, and constants $\kappa_c\in\R,\kappa_d\in\R_{>0}$, such that the following hold
		\begin{itemize}
			\item $\forall x,\hat x\in \R^n$
			\begin{align}\label{c1}
			\underline{\alpha} (\Vert x-\hat{x}\Vert ) \leq V(x,\hat{x})\leq \overline{\alpha} (\Vert x-\hat{x}\Vert ).
			\end{align}
			\item $\forall x,\hat x\in \R^n~\text{a}.\text{e}$, and $\forall u,\hat u\in \mathbb{U}$
			\begin{align}\label{c2}
			\dfrac{\partial V(x,\hat{x})}{\partial x} f(x,u)+\dfrac{\partial V(x,\hat{x})}{\partial \hat{x}} f(\hat x,\hat{u})
			\leq -\kappa_c V(x,\hat{x})+\rho_{u_c}(\Vert u- \hat{u}\Vert ) .
			\end{align}
			\item $\forall x,\hat x\in \R^n$ and $\forall u,\hat u\in \mathbb{U}$
			\begin{align}\label{c3}
			V(g(x,u),g(\hat x,\hat{ u}))
			\leq {\kappa_d} V(x,\hat{x})+\rho_{u_d}(\Vert u - \hat{u}\Vert ).
			\end{align}
		\end{itemize} 
	\end{assumption}
	\begin{assumption}\label{ass2} 
		There exists a $\mathcal{K}_{\infty}$ function $\hat{\gamma}$ such that
		\begin{align}\label{tinq} 
		\forall x,y,z \in \R^n,~~V(x,y)\leq V(x,z)+\hat{\gamma}(\Vert y-z\Vert).
		\end{align}
	\end{assumption} 
	\begin{remark} Assumption \ref{likeiss} has different implications based on the values of $\kappa_c$ and $\kappa_d$ as the following. Given \eqref{c1} holds: $(i)$ the existence of function $V$ satisfying \eqref{c2} and \eqref{c3} with $k_c\leq 0$ and $k_d\ge 1$ results in incremental forward completeness of the continuous and discrete dynamics of $\Sigma$, respectively, and we say $\Sigma_c$ and $\Sigma_d$ are $\delta$-FC \cite{6082386}; $(ii)$ the existence of function $V$ satisfying \eqref{c2} and \eqref{c3} with $k_c>0$ and $k_d< 1$ results in incremental input-to-state stability of the continuous and discrete dynamics of $\Sigma$, respectively,  and we say $\Sigma_c$ and $\Sigma_d$ are $\delta$-ISS \bb{\cite{angeli09,ruffer}}. 
		In addition, Assumptions \ref{ass2} is non-restrictive conditions provided that one is interested to work on a compact subset of $\R^n$ \cite{zamani2014symbolic}.  	\hfill$\diamond$
	\end{remark} 
	\begin{remark}
		In condition \eqref{c2}, `` $\forall x,\hat x\in\R^n ~ \text{a}.\text{e}.$" should be interpreted as ``for every $x,\hat x\in \R^n$ except on a set of zero Lebesgue-measure in $\R^n$". From Rademacher’s theorem \cite{federer}, the local Lipschitz assumption on function $V$ ensures that $\dfrac{\partial V(x,\hat{x})}{\partial x} f(x,u)+\dfrac{\partial V(x,\hat{x})}{\partial \hat{x}} f(\hat x,\hat{u})$ is well defined, except on a set of measure zero.	\hfill$\diamond$
	\end{remark}
	
	The following lemma provides \bb{a bound to} the evolution of function $V$ in Assumption \ref{likeiss} which is needed in the proof of Theorem \ref{thm1}.
	\begin{lemma}\label{lemma1}
		Consider impulsive system $\Sigma=(\R^n,\mathbb U,\mathcal{U}_{\tau},f,g)$ with jump parameters $\tau$, $p_1$ and $p_2$, where $\mathcal{U}_{\tau}$ is given by \eqref{input}. Let \eqref{c2} in Assumption \ref{likeiss} holds. Then
		for all $x,\hat x\in \R^n$, for all $\nu,\hat \nu\in \mathcal{U}_{\tau}$, and for any two consecutive impulses $(t_{k},t_{k+1})$, one has  
		\begin{align}\label{c4}
		V(\mathbf{x}_{x,\nu}(\pre{-}t_{k+1}),\mathbf{x}_{\hat{x},\hat{\nu}}(\pre{-}t_{k+1}))
		\leq e^{-\kappa_c(t_{k+1}-t_{k})} V(\mathbf{x}_{x,\nu}(t_{k}),\mathbf{x}_{\hat{x},\hat{\nu}}(t_{k}))
		+\frac{1-e^{-\kappa_c(t_{k+1}-t_{k})}}{\kappa_c}\rho_{u_c}(\Vert \nu- \hat{\nu}\Vert_{\infty}).
		\end{align}
	\end{lemma}
	\begin{proof}
		By using Lemma $1$ in \cite{Teel}, condition \eqref{c2} implies that between any
		two consecutive impulses $ (t_{k},t_{k+1})$, function $V$ is absolutely continuous and satisfies 
		\begin{align*}
		\dot{V}(\mathbf{x}_{x,\nu}(t),\mathbf{x}_{\hat{x},\hat{\nu}}(t))
		\leq-\kappa_c V(\mathbf{x}_{x,\nu}(t),\mathbf{x}_{\hat{x},\hat{\nu}}(t))+\rho_{u_c}(\Vert \nu(t)- \hat{\nu}(t)\Vert)
		\end{align*}
		$ \forall t\in [t_{k},t_{k+1})~\text{a}.\text{e}$. Hence, it follwos from \cite[Theorem $3.1$]{Szarski} that 
		\begin{align*}
		V(\mathbf{x}_{x,\nu}(\pre{-}t_{k+1}),\mathbf{x}_{\hat{x},\hat{\nu}}(\pre{-}t_{k+1}))
		\leq& e^{-\kappa_c(t_{k+1}-t_{k})} V(\mathbf{x}_{x,\nu}(t_{k}),\mathbf{x}_{\hat{x},\hat{\nu}}(t_{k}))
		+\int_{t_{k}}^{t_{k+1}}e^{-\kappa_c(t_{k+1}-s)}\rho_{u_c}(\Vert \nu(s)- \hat{\nu}(s)\Vert)ds\\
		\leq&e^{-\kappa_c(t_{k+1}-t_{k1})} V(\mathbf{x}_{x,\nu}(t_{k}),\mathbf{x}_{\hat{x},\hat{\nu}}(t_{k}))
		+\frac{1-e^{-\kappa_c(t_{k+1}-t_{k})}}{\kappa_c}\rho_{u_c}(\Vert \nu- \hat{\nu}\Vert_{\infty}).
		\end{align*}
	\end{proof}	
	
	We now have all the ingredients to construct a symbolic model $\hat{T}_{\tau}(\Sigma)$ of transition system $T_{\tau}(\Sigma)$ associated to the impulsive system $\Sigma$ admitting a function $V$ that satisfies Assumption \ref{likeiss} as follows.
	\begin{definition}\label{smm} Consider a transition system $T_{\tau}(\Sigma)=(X,X_0,U,\TF,Y,\op)$, associated to the impulsive system $\Sigma=(\R^n,\mathbb U,\mathcal{U}_{\tau},f,g)$. Assume $\Sigma$ admits a function $V$ that satisfies Assumption \ref{likeiss}. Then one can construct symbolic model $\hat T_{\tau}(\Sigma)=(\hat{X},\hat{X}_0,\hat{U},\hat{\TF},\hat{Y},\hat{\op})$ where:
		\begin{itemize}\addtolength{\itemindent}{-0.4cm}
			\item $\hat{X}=\hat{\R}^n\times \{0,\cdots,p_2\}$, where $\hat{\R}^n=[\R^n]_{\eta}$ and $\eta$ is the state space quantization parameter; 
			\item $\hat{X}_0=\hat{\R}^n\times \{0\}$;
			\item $\hat{U}=[\mathbb U]_{\mu}$, where $\mu$ is the input set quantization parameter;
			\item $(\hat{x}^+,l^+)\in \hat{\TF}((\hat{x},l),\hat{u})$ if and only if one of the following scenarios hold:
			\begin{itemize}\addtolength{\itemindent}{-0.95cm} 
				\item Flow scenario: $0\leq l\leq p_{2}-1$, $\Vert\hat{x}^+- \mathbf x_{\hat{x},\hat{u}}(\pre{-}\tau)\Vert\leq\eta$, $l^+=l+1$;
				\item Jump scenario: $p_1\leq l\leq p_2$, $\Vert\hat{x}^+- g(\hat x,\hat{u})\Vert\leq\eta$, $l^+=0$;
			\end{itemize}
			\item $\hat{Y}=Y$;
			\item $\hat{\op}:\hat{X}\rightarrow \hat{Y}$, defined as $\hat{\op}(\hat x,l)=\hat x$.
		\end{itemize} 
	\end{definition}
	An illustration of the computation of the transitions of $\hat T_{\tau}(\Sigma)$ is shown in Figure \ref{trans}.
	In the definition
	of the transition function, and in the remainder of the paper, we abuse
	notation by identifying $\hat u$ with the constant input curve with domain
	$[0,\tau)$ and value $\hat u$.
	\begin{figure}[t]
		\centering
		\subfigure[]{\includegraphics[height=4.0cm, width=4.5cm]{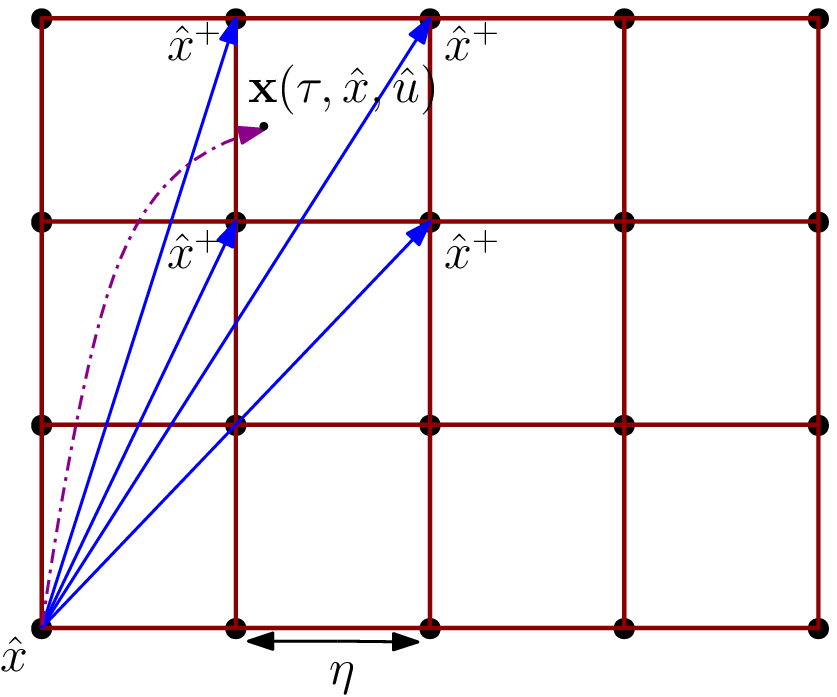}}
		\subfigure[]{\includegraphics[height=4.0cm, width=4.5cm]{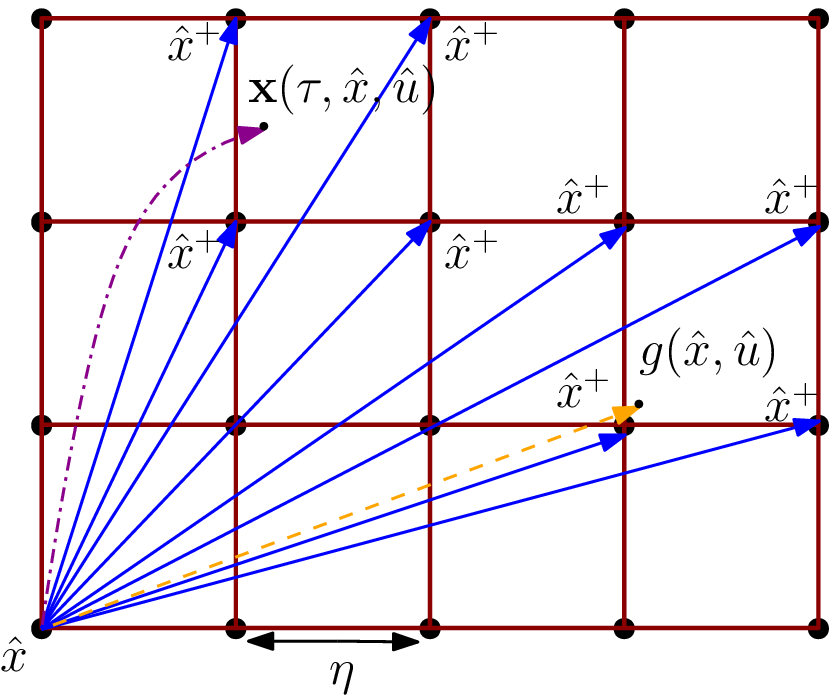}}
		\subfigure[]{\includegraphics[height=4.0cm, width=4.5cm]{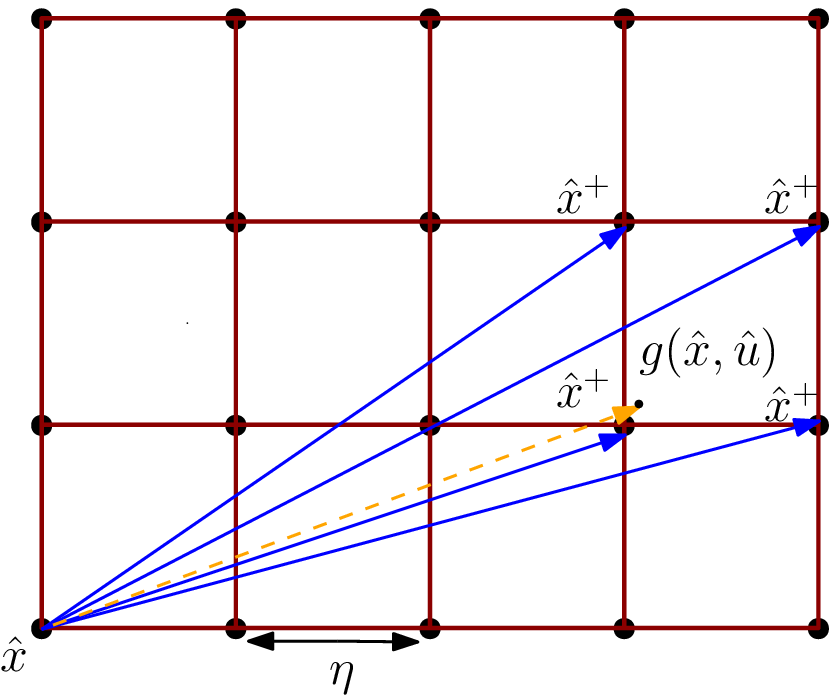}}
		\vspace{-0.3cm}
		\caption{ An illustration of the computation of the transitions of $\hat T_{\tau}(\Sigma)$ for particular $\hat x$ and $\hat u$ with (a) $l<p_1$, (b) $p_1\leq l\leq p_{2}-1$, and (c) $l=p_2$.}
		\label{trans}
	\end{figure}
	Now, we establish the relation between $T_{\tau}(\Sigma)$ and $\hat T_{\tau}(\Sigma)$, introduced above, via the notion of alternating simulation function as in Definition \ref{sf}.
	\begin{theorem}\label{thm1}
		Consider an impulsive system $\Sigma=(\R^n,\mathbb U,\mathcal{U}_\tau,f,g)$ with its associated transition system $T_{\tau}(\Sigma)=(X,X_0,U,\TF,Y,\op)$. Let Assumptions \ref{likeiss}, and \ref{ass2} hold. Consider symbolic model $\hat T_{\tau}(\Sigma)=(\hat{X},\hat{X}_0,\hat{U},\hat{\TF},\hat{Y},\hat{\op})$ constructed as in Definition \ref{smm}. If inequality
		\begin{align}\label{cc1}
		\ln(\kappa_{d})-\kappa_c\tau l<0
		\end{align}
		holds for $l\in\{p_1,p_2\}$, then function $\mathcal{V}$ defined as 
		\begin{align}\label{sm}
		\mathcal{V}((x,l),(\hat{x},l))\Let\left\{
		\begin{array}{lr} 
		V(x,\hat{x})    \quad\quad\quad\,~ \text{if}\quad \kappa_{d}<1 ~\&~ \kappa_c>0,\\
		V(x,\hat{x})e^{\kappa_c\tau \epsilon l}   \quad \text{if}\quad \kappa_{d}\geq1 ~\&~ \kappa_c>0,\\
		V(x,\hat{x})\kappa_{d}^{\frac{l}{\delta}} \quad\quad\, \text{if}\quad \kappa_{d}<1 ~\&~ \kappa_c\leq0,
		\end{array}\right.
		\end{align}
		for some $0<\epsilon<1$ and $\delta>p_2$, is an alternating simulation function from $\hat T_{\tau}(\Sigma)$ to $T_{\tau}(\Sigma)$. 
	\end{theorem}
	
	After proving Theorem \ref{thm1}, we will provide additional insight into
	condition \eqref{cc1}. Note that for the case in which $\kappa_{d}\ge1$ and $\kappa_c\leq0$, this condition cannot hold at all. Hence this case is excluded from the definition of $\mathcal{V}$ in \eqref{sm}. 
	\begin{proof} By using \eqref{c1}, $\forall (x,l)\in X$ and $ \forall (\hat{x},l) \in \hat{X}
		$, we have 
		\begin{align*}
		\Vert\op(x,l)-\hat{\op}(\hat{x},l)\Vert=\Vert x-\hat x\Vert&\leq 
		\underline{\alpha}^{-1}(V(x,\hat{x}))\leq\hat{\alpha}\left(\mathcal{V}((x,l),(\hat{x},l))\right),
		\end{align*}
		where
		\begin{align*}
		\hat{\alpha}(s)=\left\{
		\begin{array}{lr} 
		\underline{\alpha}^{-1}(s)    \quad\quad\quad\quad\,\,\, \text{if}\quad \kappa_{d}<1 ~\&~ \kappa_c>0,\\
		\underline{\alpha}^{-1}(e^{-\kappa_c\tau \epsilon p_1} s)  \quad \text{if}\quad \kappa_{d}\geq1 ~\&~ \kappa_c>0,\\
		\underline{\alpha}^{-1}(\kappa_{d}^{-\frac{p_2}{\delta}}s) \quad\quad\,\, \text{if}\quad \kappa_{d}<1 ~\&~ \kappa_c\leq0,
		\end{array}\right.
		\end{align*}
		for all $s\in\R_{\geq0}$. Hence, \eqref{sf1} holds with $\tilde{\alpha}=\hat{\alpha}^{-1}$. 
		
		Now we show that inequality \eqref{sf2} holds as well. Consider any $\hat u\in \hat{U}$ and choose $\nu(\cdot)=\hat{u}$. Then, using \eqref{tinq}, for all $x,\hat x\in \R^n$, for all $\hat{u}\in \hat{U}$, we have in the flow scenario the following inequality:
		\begin{align*}
		V(&\mathbf{x}_{x,\hat{u}}(\pre{-}\tau),\hat{x}^+)\leq V(\mathbf{x}_{x,\hat{u}}(\pre{-}\tau),\mathbf{x}_{\hat{x},\hat{u}}(\pre{-}\tau))+\hat{\gamma}(\Vert \hat{x}^+-\mathbf{x}_{\hat{x},\hat{u}}(\pre{-}\tau)\Vert).
		\end{align*}		
		Now, from Definition \ref{smm}, the above inequality reduces to
		\begin{align*}
		V(\mathbf{x}_{x,\hat{u}}(\pre{-}\tau),\hat{x}^+)&
		\leq V(\mathbf{x}_{x,\hat{u}}(\pre{-}\tau),\mathbf{x}_{\hat{x},\hat{u}}(\pre{-}\tau))+\hat{\gamma}(\eta),
		\end{align*}
		for any $\hat{x}^+$ such that  $(\hat{x}^+,l^+)\in \hat{\TF}((\hat{x},l),\hat{u})$.
		From \eqref{c4} with $t_{k+1}=\tau, t_{k}=0$, one gets 
		\begin{align*}
		V(\mathbf{x}_{x,\hat{u}}(\pre{-}\tau),\mathbf{x}_{\hat{x},\hat{u}}(\pre{-}\tau))
		&\leq e^{-\kappa_c\tau} V(\mathbf{x}_{x,\hat u}(0),\mathbf{x}(0,\hat{x},\hat{u}))=e^{-\kappa_c\tau} V(x,\hat{x})
		\end{align*}
		Hence, for all $x,\hat x\in \R^n$, for all $\hat{u}\in \hat{{U}}$, one obtains
		\begin{align}\label{cv}
		V(\mathbf{x}_{x,\hat{u}}(\pre{-}\tau),\hat{x}^+)&
		\leq e^{-\kappa_c\tau} V(x,\hat{x})+\hat{\gamma}(\eta),
		\end{align}for any $\hat{x}^+$ such that  $(\hat{x}^+,l^+)\in \hat{\TF}((\hat{x},l),\hat{u})$.
		By following similar argument to the previous one and using \eqref{c3}, one also obtains the following inequality in the jump scenario for all $x,\hat x\in \R^n$, and for all $\hat{u}\in \hat{{U}}$
		\begin{align}\label{1cv}
		V(g(x,\hat{u}),\hat{x}^+)&
		\leq \kappa_d  V(x,\hat{x})+\hat{\gamma}(\eta),
		\end{align}for any $\hat{x}^+$ such that  $(\hat{x}^+,l^+)\in \hat{\TF}((\hat{x},l),\hat{u})$.
		
		Now, in order to show function $\mathcal{V}$ defined in \eqref{sm} satisfies \eqref{sf2}, we consider the different scenarios in Definition \ref{smm} and different cases for values of $\kappa_{d}$ and $\kappa_c$ as follows:
		
		\begin{itemize}
			\item $\kappa_{d}<1$ $\&$ $\kappa_c>0$ (\textbf{case 1}):
			\begin{itemize}
				\item Flow scenario $(l^+=l+1)$:
				\begin{align*}
				\mathcal{V}((x^+,l^+),(\hat{x}^+,l^+))&=V(x^+,\hat{x}^+)\leq e^{-\kappa_c\tau} V(x,\hat{x})+\hat{\gamma}(\eta)=e^{-\kappa_c\tau} \mathcal{V}((x,l),(\hat{x},l))+\hat{\gamma}(\eta).
				\end{align*}
				\item Jump scenario $(l^+=0)$:
				\begin{align*}
				\mathcal{V}((x^+,l^+),(\hat{x}^+,l^+))&=V(x^+,\hat{x}^+)\leq \kappa_{d} V(x,\hat{x})+\hat{\gamma}(\eta)=\kappa_{d} \mathcal{V}((x,l),(\hat{x},l))+\hat{\gamma}(\eta).
				\end{align*}
			\end{itemize}
			Let $\lambda_{f}=\max\{e^{-\kappa_c\tau},\kappa_{d}\}$, and $\gamma_{f}=\hat{\gamma}$, then  
			\begin{align*}
			\mathcal{V}((x^+,l^+),(\hat{x}^+,l^+))
			\leq\lambda_{f}\mathcal{V}((x,l),(\hat{x},l))+\gamma_{f}(\eta).
			\end{align*} 
			\item $\kappa_{d}\ge1$ $\&$ $\kappa_c>0$ (\textbf{case 2}):
			\begin{itemize}
				\item Flow scenario $(l^+=l+1)$:
				\begin{align*}
				&\mathcal{V}((x^+,l^+),(\hat{x}^+,l^+))=V(x^+,\hat{x}^+)e^{\kappa_c\tau \epsilon l^+}=V(x^+,\hat{x}^+)e^{\kappa_c\tau \epsilon (l+1)}\leq(e^{-\kappa_c\tau} V(x,\hat{x})+\hat{\gamma}(\eta))e^{\kappa_c\tau \epsilon (l+1)}\\
				&=e^{-\kappa_c\tau}e^{\kappa_c\tau \epsilon}e^{\kappa_c\tau \epsilon l} V(x,\hat{x})+\frac{\hat{\gamma}(\eta)}{e^{-\kappa_c\tau \epsilon (l+1)}}=e^{-\kappa_c\tau(1-\epsilon)}\mathcal{V}((x,l),(\hat{x},l))+\frac{\hat{\gamma}(\eta)}{e^{-\kappa_c\tau \epsilon (l+1)}}.
				\end{align*}
				\item Jump scenario $(l^+=0)$:
				\begin{align*}
				&\mathcal{V}((x^+,l^+),(\hat{x}^+,l^+))=V(x^+,\hat{x}^+)e^{\kappa_c\tau \epsilon l^+}=V(x^+,\hat{x}^+)\leq\kappa_d V(x,\hat{x})+\hat{\gamma}(\eta)\\
				&=\frac{e^{\kappa_c\tau \epsilon l}}{e^{\kappa_c\tau \epsilon l}}\kappa_d V(x,\hat{x})+\hat{\gamma}(\eta)=e^{-\kappa_c\tau \epsilon l}\kappa_{d}\mathcal{V}((x,l),(\hat{x},l))+\hat{\gamma}(\eta). 
				\end{align*}
			\end{itemize}
			Let  $\lambda_{f}=\max\{e^{-\kappa_c\tau(1-\epsilon)},e^{-\kappa_c\tau \epsilon p_1}\kappa_{d}\}$, and $\gamma_{f}=e^{\kappa_c\tau \epsilon (p_2+1)}\hat{\gamma}$, then  
			\begin{align*}
			\mathcal{V}&((x^+,l^+),(\hat{x}^+,l^+))
			\leq\lambda_{f}\mathcal{V}((x,l),(\hat{x},l))+\gamma_{f}(\eta).
			\end{align*} 
			\item $\kappa_{d}<1$ $\&$ $\kappa_c\leq0$ (\textbf{case 3}):
			\begin{itemize}
				\item Flow scenario $(l^+=l+1)$:
				\begin{align*}
				&\mathcal{V}((x^+,l^+),(\hat{x}^+,l^+))=V(x^+,\hat{x}^+)\kappa_{d}^{\frac{l^+}{\delta } }=V(x^+,\hat{x}^+)\kappa_{d}^{\frac{(l+1)}{\delta }}\leq(e^{-\kappa_c\tau} V(x,\hat{x})+\hat{\gamma}(\eta))\kappa_{d}^{\frac{(l+1)}{\delta }}\\
				&=e^{-\kappa_c\tau}\kappa_{d}^{\frac{l}{\delta }}\kappa_{d}^{\frac{1}{\delta }} V(x,\hat{x})+\hat{\gamma}(\eta)\kappa_{d}^{\frac{(l+1)}{\delta }}=e^{-\kappa_c\tau}\kappa_{d}^{\frac{1}{\delta }}\mathcal{V}((x,l),(\hat{x},l))+\hat{\gamma}(\eta)\kappa_{d}^{\frac{(l+1)}{\delta }}.
				\end{align*}
				\item Jump scenario $(l^+=0)$:
				\begin{align*}
				&\mathcal{V}((x^+,l^+),(\hat{x}^+,l^+))=V(x^+,\hat{x}^+)\kappa_{d}^{\frac{l^+}{\delta }}=V(x^+,\hat{x}^+)\leq\kappa_d V(x,\hat{x})+\hat{\gamma}(\eta)\\
				&=\frac{\kappa_{d}^{\frac{l}{\delta }}}{\kappa_{d}^{\frac{l}{\delta }}}\kappa_d V(x,\hat{x})+\hat{\gamma}(\eta)=\kappa_{d}^{\frac{\delta-l}{\delta }}\mathcal{V}((x,l),(\hat{x},l))+\hat{\gamma}(\eta). 
				\end{align*}
			\end{itemize}
			Let $\lambda_{f}=\max\{e^{-\kappa_c\tau}\kappa_{d}^{\frac{1}{\delta }},\kappa_{d}^{\frac{\delta-p_2}{\delta }}\}$, and $\gamma_{f}=\hat{\gamma}$, then  
			\begin{align*}
			\mathcal{V}&((x^+,l^+),(\hat{x}^+,l^+))
			\leq\lambda_{f}\mathcal{V}((x,l),(\hat{x},l))+\gamma_{f}(\eta).
			\end{align*} 
		\end{itemize}
		To continue with the proof, we need to show that $\lambda_{f}<1$ for \textbf{case 2} and \textbf{case 3} (\textbf{case 1} is trivial).
		In \textbf{case 2}, note that $e^{-\kappa_c\tau(1-\epsilon)}<1$ since  $0<\epsilon<1$ and $\kappa_c>0$. Additionally, $e^{-\kappa_c\tau \epsilon p_1}\kappa_{d}<1 \Leftrightarrow \ln(\kappa_{d})-\kappa_c\tau\epsilon p_1<0$. By continuity of the real number
		, we can
		always find some $0<\epsilon<1$  such that $\ln(\kappa_{d})-\kappa_c\tau l<0$,  $l\in\{p_1,p_2\}$, implies $\ln(\kappa_{d})-\kappa_c\tau \epsilon p_1<0$. Hence, $\lambda_{f}<1$.
		Similarly, in \textbf{case 3}, we have $\kappa_{d}^{\frac{\delta-p_2}{\delta }}<1$ since  $\delta>p_2$ and $\kappa_d<1$. Moreover, 
		$e^{-\kappa_c\tau}\kappa_{d}^{\frac{1}{\delta }}<1\Leftrightarrow\ln(\kappa_{d})-\kappa_c\tau \delta <0$. By continuity of the real number, we can
		always find some $\delta>p_2$ such that $\ln(\kappa_{d})-\kappa_c\tau l<0$,  $l\in\{p_1,p_2\}$, implies $\ln(\kappa_{d})-\kappa_c\tau \delta<0$. Therefore, $\lambda_{f}<1$. 
			Hence, for all $((x,l),(\hat{x},l))\in X\times\hat{X}$, for all $\hat{u}\in \hat{U}$, for any $\hat{x}^+$ such that $(\hat{x}^+,l^+)\in \hat{\TF}((\hat{x},l),\hat{u})$, $\mathcal{V}$ satisfies inequality \eqref{sf2}	with $\nu=\hat{u}$, $\tilde{\sigma}={\lambda}_{f}$, $\tilde{\varepsilon}={\gamma}_{f}(\eta)$, and $\tilde{\rho}_{u}=0$. Thus, $\mathcal{V}$ is an alternating simulation function from $\hat T_{\tau}(\Sigma)$ to $T_{\tau}(\Sigma)$.

	\end{proof}	
	
	\begin{remark}\label{biss}
		One can also verify that function $\mathcal{V}$ given by \eqref{sm} is also an alternating simulation function from $T_{\tau}(\Sigma)$ to $\hat T_{\tau}(\Sigma)$. In particular, $\mathcal{V}$ satisfies \eqref{sf1} and \eqref{sf2} with choosing $\hat{u}$ satisfying\footnote{By the structure of $\hat{{U}}$, there always exists $\hat{u}$ satisfying $\Vert\hat{u}-\nu\Vert\leq\mu$.} $\Vert\hat{u}-\nu\Vert\leq\mu$, same $\tilde{\sigma}$, $\tilde{\rho}_{u}$ defined in Theorem \ref{thm1},  $\tilde{\varepsilon}={\gamma}_{f}(\eta)+\max\left\{\frac{1-e^{-\kappa_c(t_{k+1}-t_{k})}}{\kappa_c}\rho_{u_c},\rho_{u_d}\right\}(\mu)$ for \textbf{case 1} and \textbf{3}, and  $\tilde{\varepsilon}={\gamma}_{f}(\eta)+\max\left\{e^{\kappa_c\tau \epsilon (p_2+1)}\frac{1-e^{-\kappa_c(t_{k+1}-t_{k})}}{\kappa_c}\rho_{u_c},\rho_{u_d}\right\}(\mu)$ for \textbf{case 2}. \bb{Observe that the existence of a function $\mathcal V$ serving as an alternating simulation function in both directions, i.e. from $T_{\tau}(\Sigma)$ to $\hat T_{\tau}(\Sigma)$ and from $\hat T_{\tau}(\Sigma)$ to $ T_{\tau}(\Sigma)$, implies the existence of an approximate alternating bisimulation relation between $T_{\tau}(\Sigma)$ and $\hat T_{\tau}(\Sigma)$ as introduced in \cite{pt09}.} Consequently, $\hat T_{\tau}(\Sigma)$ is a complete symbolic model for $T_{\tau}(\Sigma)$. The completeness of the symbolic model implies that there exists a controller enforcing the desired specifications on the symbolic model $\hat T_{\tau}(\Sigma)$ \emph{if and only if} there exists a controller enforcing the same specifications on $T_{\tau}(\Sigma)$.
		\hfill$\diamond$
	\end{remark}  
	
	\begin{remark}\label{finite}
		The symbolic model $\hat T_{\tau}(\Sigma)$ has a countably infinite set of states. However, in practical applications, the physical variables are restricted to a compact set. Hence, we are usually interested in the dynamics of the impulsive system only on a compact subset $\mathbb{X}\subseteq\R^n$. Then, we can restrict the set of states of $\hat T_{\tau}(\Sigma)$ to the sets $([\R^n]_{\eta}\cap\mathbb{X})\!\times\! \{0,\cdots,p_2\}$ which is finite. We refer the interested readers to the explanation provided after Remark 4.1 in \cite{6082386} for more details. 
		\hfill$\diamond$
	\end{remark} 
	
	Finally, we would like to provide a discussion on condition \eqref{cc1} in Theorem \ref{thm1}. In the case when $\kappa_{d}<1$ and $\kappa_c>0$, the continuous and discrete dynamics of $\Sigma$ are $\delta$-ISS, and, clearly, \eqref{cc1} always holds. For the case when $\kappa_c>0$ and $\kappa_{d}\ge1$, the continuous dynamic $\Sigma_c$ is $\delta$-ISS while the discrete dynamic $\Sigma_d$ is $\delta$-FC. In order for condition \eqref{cc1} to hold in this case, $\kappa_c$ should be large enough to accommodate the undesirable effect of $\kappa_{d}$ and that the impulses do not happen too frequently. Finally, $\kappa_c\leq0$ and $\kappa_{d}<1$ corresponds to the case that the continuous dynamic $\Sigma_c$ is $\delta$-FC while the discrete one $\Sigma_d$ is $\delta$-ISS. Here, we require the impulses to happen very often and $\kappa_{d}$ to be small enough to accommodate the undesirable effect of $\kappa_{c}$. 
	\bb{Note that condition \eqref{cc1} ensures that an
		increase in the value of function $V$ in Assumption \ref{likeiss} during flows is compensated by
		a decrease at jumps and vice versa. A similar argument was used in \cite[Sections 4,5,6]{Teel} to reason about input-to-state stability of impulsive systems, and we expect that by utilizing Assumption \ref{likeiss} with condition \eqref{cc1}, one can get $\delta$-ISS for system $\Sigma$ in \eqref{eq:2}.}     
	
	%
	\section{Case study: A storage-delivery process model}\label{sec:Examples}
	
	In this case study, we apply our approach to a variant of the storage-delivery process model from \cite{DASHKOVSKIY2016}.
	Let the number $x\in\R_{\geq0} $ of goods in a
	storage be continuously evolving proportionally to the
	number of items with rate coefficient $a$. At every time instant
	$t\in\Omega=\{t_k\}_{k\in\N}$, with $t_{k+1}-t_{k}\in\{p_1\tau,\ldots,p_2\tau\}$ for a fixed jump parameters $\tau\in\R_{>0}$ and $p_1,p_2\in \N_{\ge1}$, $p_1\le p_2$, a truck comes to the storage and delivers $(b-1)\%$, or picks up $(1-b)\%$ of the current items. Let $c$ denote the number of items per time unit that can be added, through lineside delivery from the factory to the storage, or taken out, from the storage to other locations during $t\in(t_{k+1},t_{k})$. Similarly, let $d$ be the number of items that can be added, or taken out, from the storage at time instants $t\in\Omega$. The delivery and picking-up process is controlled by the input $\nu(t)=\nu(0)\in \{-1,0,1\},t\in [0,\tau)$. The evolution of
	this process can be modeled as
	\begin{align}\label{ex}
	\Sigma:\left\{
	\begin{array}{rl}
	\mathbf{\dot{x}}(t)&= a\mathbf{{x}}+c\nu(t),\quad\quad\quad\quad t\in\R_{\geq0}\backslash \Omega,\\
	\mathbf{x}(t)&= b\mathbf{x}(\pre{-}t)+d\nu(t),\quad\quad\,   t\in \Omega.
	\end{array}
	\right.
	\end{align}
	\begin{figure}[t]
		\centering
		\subfigure[]{\includegraphics[height=4.0cm, width=4.5cm]{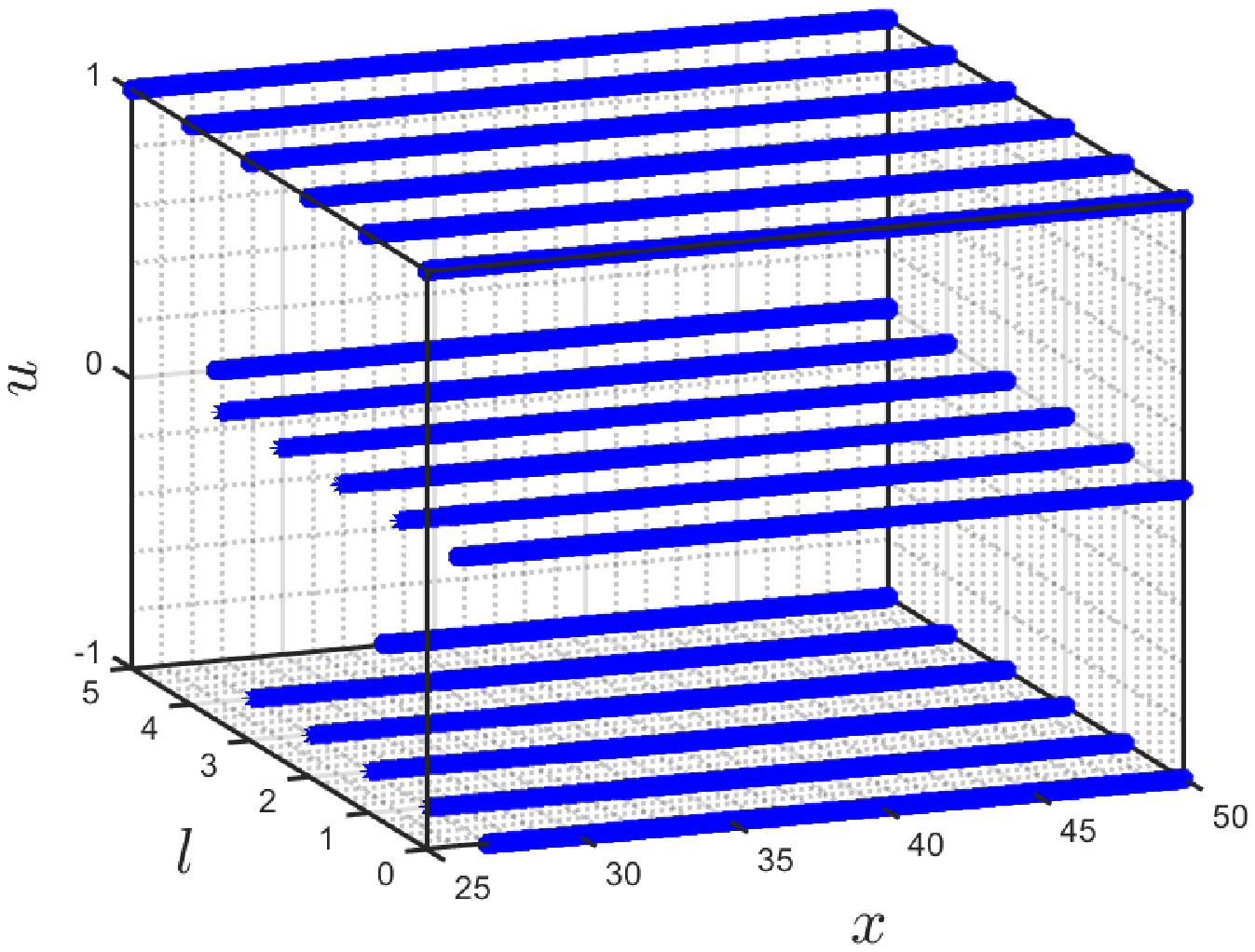}}
		\subfigure[]{\includegraphics[height=4.0cm, width=4.5cm]{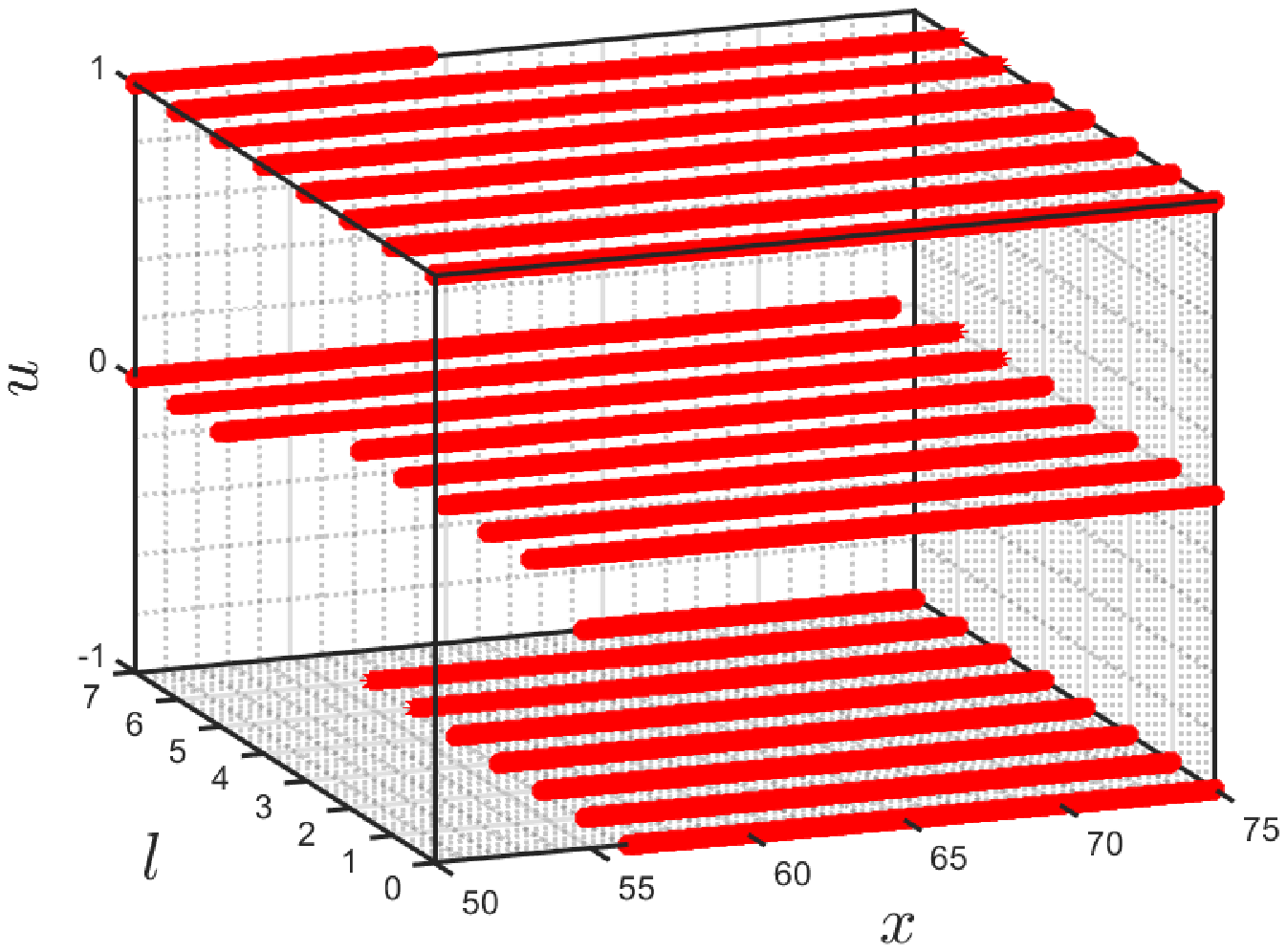}}
		\subfigure[]{\includegraphics[height=4.0cm, width=4.5cm]{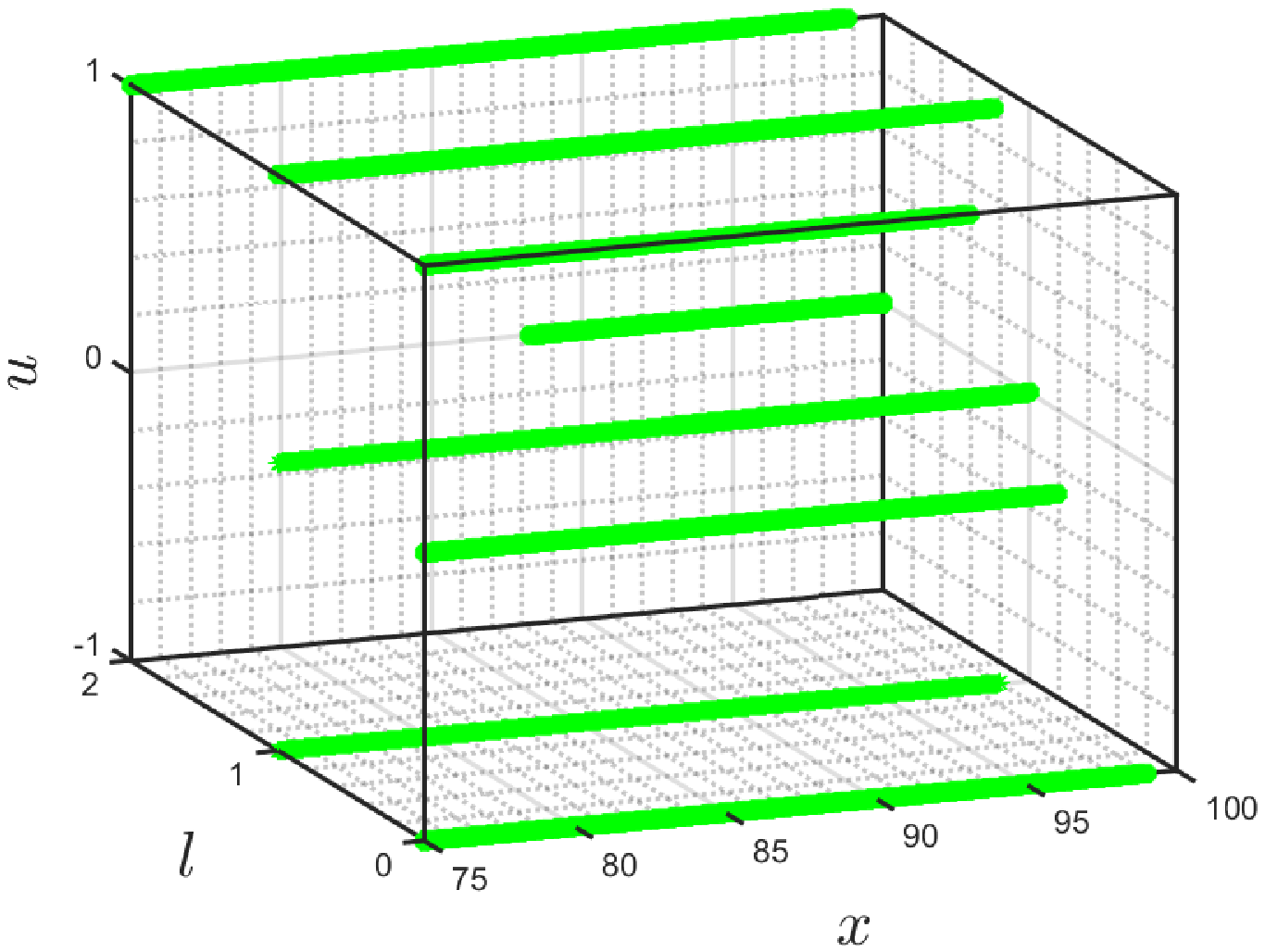}}
		\caption{Controllers with their corresponding domains: (a) \textbf{case 1}, (b) \textbf{case 2}, (c) \textbf{case 3}.}
		\label{domain}
	\end{figure}
	In order to construct a symbolic model for impulsive system $\Sigma$, we start by checking Assumptions \ref{likeiss} and \ref{ass2}. It can be shown that conditions \eqref{c1}, \eqref{c2}, and \eqref{c3} hold with $V(x,x')=\Vert x-x'\Vert$ with $\underline{\alpha}= \overline{\alpha}=\II$, $\rho_{u_c}=|c|$, $\rho_{u_d}=|d|$, $\bb{\kappa_c=-a}$, and $\kappa_d=|b|$. Moreover, condition \eqref{tinq} holds with $\hat{\gamma}=\II$. Given that \eqref{cc1} holds for $l\in\{1,p\}$, and, with a proper choice\footnote{$\epsilon=1-\varsigma$, and $\delta=p_2+\varsigma$ with $\varsigma$ sufficiently small.} of $\epsilon$ and $\delta$, function $\mathcal V(x,\hat x)$ given by \eqref{sm} is an alternating simulation function from $\hat T_{\tau}(\Sigma)$, constructed as in Definition \ref{smm}, to $T_{\tau}(\Sigma)$. In particular, $\mathcal V$ satisfies conditions \eqref{sf1} and \bb{\eqref{sf2}} with functions $\tilde{\alpha},\tilde{\rho}_u$ and constants $\tilde{\sigma}$, $\varepsilon$ given below based on the value of $a$ and $b$, with $\psi=0.99$.
	\begin{itemize}
		\item $\bb{|b|<1}$ $\&$ $\bb{a<0}$: $\tilde{\alpha}=\II$, $\tilde{\rho}_u=0$, $\tilde{\sigma}=\max\{e^{\bb{a}\tau},|b|\}$, $\tilde{\varepsilon}=\eta$. 
		
		\item $|b|\ge1$ $\&$ $\bb{a<0}$:  $\tilde{\alpha}=e^{\bb{-a}\tau \epsilon p_1}$, $\tilde{\rho}_u=0$, $\tilde{\sigma}=\max\{e^{\bb{a}\tau(1-\epsilon)},e^{\bb{a}\tau \epsilon p_1}|b|\}$, and $\tilde{\varepsilon}=e^{\bb{-a}\tau \epsilon (p_2+1)}\eta$.
		
		\item $|b|<1$ $\&$ $\bb{a\geq0}$:  $\tilde{\alpha}=|b|^{\frac{p_2}{\delta}}$, $\tilde{\rho}_u=0$, $\tilde{\sigma}=\max\{e^{\bb{a}\tau}|b|^{\frac{1}{\delta }},|b|^{\frac{\delta-p_2}{\delta }}\}$, and $\tilde{\varepsilon}=\eta$.
	\end{itemize}
	\begin{figure}[h]
		\centering
		\includegraphics[height=4.5cm, width=9cm]{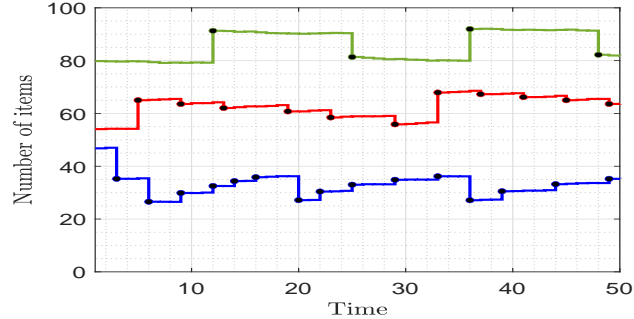}
		\caption{Trajectories of system $\bb{\Sigma}$ for different values of $p_1$, $p_2$, $a$, $b$,$c$, $d$, $\Psi$: \textcolor{blue}{blue} \bb{(bottom)} (\bb{$p_1=1$, $p_2=5$, $a=-0.2$, $b=0.9$, $c=d=5$, $\Psi=\{25,50\}$}), \textcolor{red}{red} \bb{(middle)} (\bb{$p_1=5$, $p_2=7$, $a=-0.3$, $b=1.01$, $c=d=15$, $\Psi=\{50,75\}$}), \textcolor{myco}{green} \bb{(top)} (\bb{$p_1=1$, $p_2=2$, $a=0.2$, $b=0.85$, $c=d=15$, $\Psi=\{75,100\}$}). The jumps are indicated by {\scriptsize $\bullet$}.}
		\label{st}
	\end{figure}  
	The control objective here is to maintain the number of items in a desired range $\Psi$ given by $\Psi=[\psi_l,\psi_u]$ (a safety specification). For \bb{the} sake of numerical illustration, we choose different \bb{combinations} of $p_1$, $p_2$, $a$, $b$, $c$, $d$, $\Psi$, and leverage software tool \texttt{SCOTS} \cite{Rungger} for constructing symbolic models $\hat T_{\tau}(\Sigma)$ and controller $u$ for $T_{\tau}(\Sigma)$ with $\tau=0.2$, and $\eta=0.01$. \bb{The controllers for all cases with their domains are represented on Figure \ref{domain}: (a) \textbf{case 1}, (b) \textbf{case 2}, (c) \textbf{case 3}.}
	In addition, Figure \ref{st} shows trajectories of system $\Sigma$ for different values of $p_1$, $p_2$, $a$, $b$, $c$, $d$, $\Psi$ as follows: \textbf{case 1} (blue) \bb{(bottom)} : \bb{$p_1=1$, $p_2=5$, $a=-0.2$, $b=0.9$, $c=d=10$, $\Psi=\{25,50\}$}; \textbf{case 2} (red) \bb{(middle)}: \bb{$p_1=5$, $p_2=7$, $a=-0.3$, $b=1.01$, $c=d=15$, $\Psi=\{50,75\}$}; \textbf{case 3} (green)\bb{(top)}: \bb{$p_1=1$, $p_2=2$, $a=0.2$, $b=0.85$, $c=d=15$, $\Psi=\{75,100\}$}. Finally, one can compute the mismatch between the output behavior of  $T_{\tau}(\Sigma)$ and its symbolic model $\hat T_{\tau}(\Sigma)$ by utilizing Proposition \ref{error}. In \bb{particular}, we have $\hat{\varepsilon}=\bb{0.25}$ for \textbf{case 1}, $\hat{\varepsilon}=\bb{0.75}$ for \textbf{case 2}, and $\hat{\varepsilon}=\bb{0.65}$ for \textbf{case 3}. 
	\section{Conclusion}
	In this work, we provided an approach for constructing symbolic models of impulsive systems. To do so, we used a notion of alternating simulation functions to relate impulsive systems and their symbolic models. Under some stability properties, we introduced an approach to construct symbolic models for a class of impulsive systems. Finally, we illustrated the effectiveness of our results via a model of storage-delivery-process. 
	
	\bibliographystyle{ieeetr}      
	\bibliography{arxiv} 

\end{document}